\documentclass[10pt,conference,twocolumn]{ieeeconf}

\usepackage[utf8]{inputenc}
\usepackage[T1]{fontenc}
\usepackage{amsmath,amssymb}
\let\ieeeproof\proof
\let\endieeeproof\endproof
\let\proof\relax
\let\endproof\relax
\usepackage{amsthm}
\let\proof\ieeeproof
\let\endproof\endieeeproof
\usepackage{mathtools}
\usepackage{bm}
\usepackage{hyperref}
\usepackage{algorithm}
\usepackage{algorithmic}
\usepackage{graphicx}
\usepackage{color}
\usepackage{cite}
\let\ieeelabelindent\labelindent
\let\labelindent\relax
\usepackage{enumitem}
\let\labelindent\ieeelabelindent
\usepackage{booktabs}
\usepackage{thmtools}

\theoremstyle{plain}
\newtheorem{theorem}{Theorem}
\newtheorem{lemma}[theorem]{Lemma}

\theoremstyle{definition}
\newtheorem{definition}{Definition}
\newtheorem{assumption}{Assumption}

\newtheorem{remark}{Remark}

\newcommand{\R}{\mathbb{R}}

\newcommand{\calH}{\mathcal{H}}

\newcommand{\calU}{\mathcal{U}}

\newcommand{\calX}{\mathcal{X}}

\newcommand{\norm}[1]{\left\| #1 \right\|}

\DeclareMathOperator{\col}{col}

\title{Stability Guarantees for Data-Driven Predictive Control of Nonlinear Systems via Approximate Koopman Embeddings}

\IEEEaftertitletext{\vspace{-1\baselineskip}}

\author{\authorblockN{Amin Taghieh and SangWoo Park}
\thanks{The authors are with the New Jersey Institute of Technology, Newark, NJ 07102 USA. (e-mails: at2228@njit.edu, sangwoo.park@njit.edu)}}

\begin{document}
\maketitle

\begin{abstract}
Data-driven model predictive control based on Willems' fundamental lemma has proven effective for linear systems, but extending stability guarantees to nonlinear systems remains an open challenge. In this paper, we establish conditions under which data-driven MPC, applied directly to input-output data from a nonlinear system, yields practical exponential stability. The key insight is that the existence of an approximate Koopman linear embedding certifies that the nonlinear data can be interpreted as noisy data from a linear time-invariant system, enabling the application of existing robust stability theories. Crucially, the Koopman embedding serves only as a theoretical certificate; the controller itself operates on raw nonlinear data without knowledge of the lifting functions. We further show that the proportional structure of the embedding residual can be exploited to obtain an ultimate bound that depends only on the irreducible offset, rather than the worst-case embedding error. The framework is demonstrated on a synchronous generator connected to an infinite bus, for which we construct an explicit physics-informed embedding with error bounds.
\end{abstract}

\vspace{-2mm}

\section{Introduction}\label{sec:intro}
\vspace{-1mm}
\subsection{Motivation and Background}
\vspace{-1mm}
Data-driven control methods that bypass explicit system identification have become a central topic in modern control theory, building on Willems' fundamental lemma~\cite{willems2005}. For linear time-invariant (LTI) systems, this lemma provides a complete, non-parametric characterization of the system's trajectory space using a single persistently exciting input-output trajectory. When combined with model predictive control (MPC), this yields the Data-Enabled Predictive Control (DeePC) algorithm \cite{coulson2019} and its variants with rigorous closed-loop stability and robustness guarantees \cite{berberich2021}.
 
Beyond MPC, Willems' fundamental lemma has led to a broad range of data-driven controller designs for LTI systems. These include computing stabilizing feedback controllers and optimal LQR gains directly from data~\cite{depersis2020}, characterizing which control-theoretic properties can be inferred from a given dataset through a ``data informativity'' framework \cite{vanwaarde2020}, and unifying direct data-driven and indirect identification-based formulations via regularization \cite{dorfler2023}. These developments provide a mature theoretical foundation for LTI systems, but all rely fundamentally on the linearity assumption: the data must come from a linear system, possibly corrupted by noise.
 
Several approaches have been proposed to extend data-driven methods to specific classes of nonlinear systems, including LTI embeddings for bilinear systems \cite{markovsky2023}, sum-of-squares programs for polynomial systems \cite{guo2022}, coordinate transformations for feedback-linearizable systems \cite{alsalti2023}, and linearization-based tracking MPC \cite{berberich2022}.

A fundamentally different perspective is offered by Koopman operator theory \cite{koopman1931}, which provides a pathway toward applying linear methods to broad classes of nonlinear systems by lifting the dynamics into a higher-dimensional (potentially infinite) functional space where they evolve linearly. This Koopman-based idea has been pursued along two distinct lines. The first is \emph{identification-based}: one learns a finite-dimensional Koopman model from data via extended dynamic mode decomposition (EDMD) \cite{korda2018} and then applies standard MPC using the identified model. Recent work has established stability guarantees for this approach under proportional approximation error bounds \cite{bold2025,schimperna2025,dejong2024}. The second line is \emph{directly data-driven}: the extended Willems' fundamental lemma of \cite{shang2024} shows that the trajectory space of a nonlinear system admitting an exact Koopman embedding can be represented directly from nonlinear data without knowledge of the lifting functions. 
 
However, a critical subtlety is that \emph{exact} finite-dimensional Koopman embeddings rarely exist for physical systems. For example, systems with trigonometric nonlinearities or bilinear terms generically require infinite-dimensional embeddings to achieve exact closure of the observables. Therefore, a fundamental question remains: under what conditions is the resulting closed loop stable, and how does the stability margin depend on the system's deviation from linearity? 
This makes the \emph{approximate} embedding framework not merely a theoretical convenience, but a practical necessity. 

In this paper, we establish rigorous conditions under which data-driven MPC, applied directly to input-output data from a nonlinear system, yields practical exponential stability of the closed loop. The key insight is that the existence of an approximate Koopman linear embedding certifies that the nonlinear data can be interpreted as noisy data from a linear time-invariant system, enabling the application of existing robust stability theories. Our contributions are as follows: (i) We formalize the notion of an \emph{approximate Koopman linear embedding with proportional-with-offset error bounds} and construct an explicit physics-informed embedding for the third-order synchronous generator; (ii) We prove that the robust data-driven MPC scheme of~\cite{berberich2021}, applied directly to input-output data from a nonlinear system, yields practical exponential stability whenever the system admits an approximate Koopman embedding with sufficiently small error; 
(iii) We provide a rigorous bridge between the Koopman approximation error and the bounded noise framework of \cite{berberich2021}, showing how the nonlinearity gap maps to an effective noise level $\bar\epsilon$ that determines the stability margin.


\vspace{-2mm}
\subsection{Notation}
For a vector $x$ and positive definite matrix $P$, a weighted norm is defined as $\norm{x}_P := \sqrt{x^\top P x}$. We write $\col(\cdot)$ for column stacking and $x_{[a,b]} := \col(x_a, \ldots, x_b)$. For a compact operating region $\calX$ with equilibrium $x_s$ and input constraint set $\calU$ with equilibrium input $u_s$, we define $\mathrm{diam}_z := \max_{x \in \calX}\norm{\Phi(x) - \Phi(x_s)}_2$ and $\mathrm{diam}_u := \max_{u \in \calU}\norm{u - u_s}_2$. Finally, a sequence $\{x_k\}_{k=0}^{N-1}$ induces the Hankel matrix
\[
\calH_L(x) := \begin{bmatrix} x_0 & x_1 & \cdots & x_{N-L} \\ x_1 & x_2 & \cdots & x_{N-L+1} \\ \vdots & & \ddots & \vdots \\ x_{L-1} & x_L & \cdots & x_{N-1}\end{bmatrix}.
\]

\vspace{-3mm}
\section{Synchronous Generator Model}\label{sec:model}
\vspace{-2mm}
In this section, we present a third-order synchronous generator model, which is nonlinear and fails to admit an exact finite-dimensional Koopman embedding. Although the theoretical results of this paper are generally applicable, we will use the generator model as a guiding example.

\vspace{-3mm}
\subsection{Continuous-Time Dynamics}
\vspace{-1mm}
In this paper, we consider the dynamics of a single synchronous generator connected to an infinite bus through a transmission line. The generator is modeled by the classical flux-decay model without exciter or governor dynamics. The state is denoted by $x = (\delta, \omega, E'_q)^\top \in \R^3$, where $\delta$ is the rotor angle, $\omega$ is the rotor angular velocity, and $E'_q$ is the $q$-axis transient internal voltage. The control input is $u = T_M \in \R$ (mechanical torque). 

\vspace{-3mm}
\begin{subequations}\label{eq:gen_explicit}
\begin{align}
&\dot{\delta} = \omega - \omega_s \label{eq:dyn_ang}\\
&M\dot{\omega} = \Big[T_M - D(\omega-\omega_s) - \frac{E'_q V_\infty}{X_\Sigma}\sin\delta\Big] \label{eq:dyn_vel}\\
&T'_{d0} \dot{E}'_q = \Big[E_{fd} - \frac{X_d+X_e}{X_\Sigma}E'_q + \frac{(X_d-X'_d)V_\infty}{X_\Sigma}\cos\delta\Big] \label{eq:dyn_Eq_prime}
\end{align}
\end{subequations}
Here, $H$ is the inertia constant, $\omega_s$ the synchronous speed, $M := \frac{2H}{\omega_s}$ the normalized inertia, $D$ the damping coefficient, $T'_{d0}$ the $d$-axis open-circuit transient time constant, $X_d$ the synchronous $d$-axis reactance, and $E_{fd}$ is the (constant) field voltage. In \eqref{eq:dyn_vel}, the term $\frac{E'_q V_\infty}{X_\Sigma}\sin\delta$ is the electrical power $P_e$, representing the electromagnetic torque. The symbol $V_\infty$ is the infinite bus voltage and $X_\Sigma := X'_d + X_e$ is the total reactance (transient $d$-axis reactance $X'_d$ plus external line reactance $X_e$). In \eqref{eq:dyn_Eq_prime}, the $\cos\delta$ term arises from the $d$-axis stator current $I_d = \frac{E'_q - V_\infty\cos\delta}{X_\Sigma}$ after substitution into the flux-decay equation. The nonlinearities are: (i) $\sin\delta$ and $\cos\delta$ (trigonometric); (ii) $E'_q\sin\delta$ (bilinear). The input $u = T_M$ enters \emph{affinely} through \eqref{eq:dyn_vel} only. 

\subsection{Discrete-Time Formulation}

With sampling period $\Delta t > 0$ and Euler discretization, define the discrete-time system:
\begin{equation}\label{eq:gen_discrete}
x_{k+1} = f(x_k, u_k), \quad y_k = g(x_k, u_k)
\end{equation}
where $f:\R^3\times\R\to\R^3$ is given component-wise by:
\begin{align}
&\delta_{k+1} = \delta_k + \Delta t(\omega_k - \omega_s) \label{eq:disc1}\\
&\omega_{k+1} = \omega_k + \frac{\Delta t}{M}\Big[u_k - D(\omega_k - \omega_s) - \frac{E'_{q,k}V_\infty}{X_\Sigma}\sin\delta_k\Big] \label{eq:disc2}\\
&E'_{q,k+1} = E'_{q,k} + \frac{\Delta t}{T'_{d0}}\Big[E_{fd} - \frac{X_d+X_e}{X_\Sigma}E'_{q,k} \notag\\
&\qquad\qquad\qquad\qquad\qquad + \frac{(X_d-X'_d)V_\infty}{X_\Sigma}\cos\delta_k\Big]. \label{eq:disc3}
\end{align} The system has $n=3$ states, $m=1$ input, and $p=2$ outputs $y_k = (\tilde\omega_k,\; P_{e,k})^\top \in \R^2$, where $\tilde\omega_k := \omega_k - \omega_s$ is the speed deviation and $P_{e,k}$ is the electrical power.


\begin{assumption}[Operating Regime]\label{ass:operating}
The generator operates in a compact domain $\calX := \{(\delta,\omega,E'_q) : |\delta - \delta_s| \leq \delta_{\max},\; |\omega - \omega_s| \leq \omega_{\max},\; E'_{q,\min} \leq E'_q \leq E'_{q,\max}\}$ where $\delta_s$ is the pre-fault equilibrium angle, $\delta_{\max} < \pi/2$ (within the stability boundary), and $\omega_{\max} \ll \omega_s$ (small speed deviations relative to synchronous speed). The input satisfies $u \in \calU := [T_{M,\min}, T_{M,\max}]$.
\end{assumption}

Under Assumption~\ref{ass:operating}, the per-step angle change satisfies $|\Delta t(\omega_k - \omega_s)| \leq \Delta t\,\omega_{\max} \ll 1$ for typical sampling rates ($\Delta t \leq 0.005$\,s), which will be exploited in the Koopman embedding analysis.

\section{Analysis of Approximate Koopman Embedding for the Generator Model}\label{sec:koopman}

\subsection{Koopman Linear Embedding: Definition and Existence}

\begin{definition}[Koopman Linear Embedding \cite{shang2024}]\label{def:koopman}
The nonlinear system \eqref{eq:gen_discrete} admits a Koopman linear embedding of dimension $n_z$ if there exist linearly independent lifting functions $\Phi = (\Phi_1,\ldots,\Phi_{n_z})^\top : \R^n \to \R^{n_z}$ such that $z_k := \Phi(x_k)$ satisfies
\vspace{-2mm}
\begin{equation}\label{eq:koopman_system}
z_{k+1} = Az_k + Bu_k, \quad y_k = Cz_k + Du_k
\end{equation}
along all trajectories of \eqref{eq:gen_discrete} within the operating domain $\calX$, with $A \in \R^{n_z\times n_z}$, $B \in \R^{n_z\times m}$, $C \in \R^{p\times n_z}$, $D \in \R^{p\times m}$.
\end{definition}

\begin{definition}[Approximate Koopman Embedding]\label{def:approx_koopman}
The system \eqref{eq:gen_discrete} admits an \emph{approximate Koopman linear embedding} centered at an equilibrium $(x_s, u_s)$ with \emph{proportional error bound} $(\epsilon_A, \epsilon_B, \epsilon_C)$ and \emph{offset} $c_0 \geq 0$ if there exist $\Phi$, $A$, $B$, $C$, $D$ such that, defining the deviation variables $\bar{z}_k := \Phi(x_k) - \Phi(x_s)$ and $\bar{u}_k := u_k - u_s$, the dynamics satisfy
\begin{equation}\label{eq:approx_koopman}
\bar{z}_{k+1} = A\bar{z}_k + B\bar{u}_k + e_k, \quad \bar{y}_k = C\bar{z}_k + D\bar{u}_k + \eta_k
\end{equation}
along all trajectories within $\calX$, where $\bar{y}_k := y_k - y_s$, and the residuals satisfy:
\begin{equation}\label{eq:proportional_bound}
\norm{e_k}_2 \leq \epsilon_A\norm{\bar{z}_k}_2 + \epsilon_B\norm{\bar{u}_k}_2 + c_0
\end{equation}
\begin{equation}\label{eq:output_bound}
\norm{\eta_k}_2 \leq \epsilon_C\norm{\bar{z}_k}_2
\end{equation}
for constants $\epsilon_A, \epsilon_B, \epsilon_C \geq 0$. When $c_0 = 0$, the bound is \emph{purely proportional}; when $c_0 > 0$, it is \emph{proportional-with-offset}. The constant affine terms arising from the equilibrium (e.g., from constant inputs like $E_{fd}$) cancel exactly by centering at the equilibrium, since $\Phi(f(x_s,u_s)) = \Phi(x_s)$.
\end{definition}

The concept of proportional error bounds for data-driven Koopman surrogates originates in the work of \cite{bold2025}, which established pointwise error bounds for kernel EDMD that are proportional in the state and input. \cite{schimperna2025} used the same structure to prove asymptotic stability of Koopman MPC with terminal conditions. The centering at equilibrium in \eqref{eq:proportional_bound} is essential: a purely proportional bound ($c_0 = 0$) guarantees \emph{exponential} stability (the error vanishes at the setpoint at an exponential rate), while a proportional-with-offset bound ($c_0 > 0$) yields \emph{practical exponential stability} (exponential convergence to a neighborhood of radius $O(c_0)$). Definition~\ref{def:approx_koopman} combines the exact embedding concept of \cite{shang2024} with the proportional error framework of \cite{bold2025,schimperna2025}, extended to include the offset term that arises naturally in finite-dimensional approximations.

\begin{remark}[Interpretation of Approximate Embedding]\label{rem:approx_interp}
In Definition~\ref{def:approx_koopman}, the lifting $\Phi$ is an explicitly chosen, exactly computed function; it is not learned or approximated. The matrices $A$, $B$, $C$, $D$ are the best linear description of how $\Phi(x_k)$ evolves under the nonlinear dynamics. The residual $e_k$ arises because the nonlinear dynamics do not propagate $\Phi$ in a perfectly linear fashion: $\Phi(f(x_k,u_k)) \neq A\Phi(x_k) + Bu_k$ exactly. Neither $\Phi$ nor $A, B$ are ``wrong''; the error is structural, reflecting the fact that the observable algebra does not close in finite dimensions. The LTI system obtained by propagating $(A, B, C, D)$ without the residuals ($e_k = 0$, $\eta_k = 0$) is referred to as the \emph{nominal Koopman system}. It is neither the true nonlinear dynamics nor a ``true'' (infinite-dimensional) Koopman operator; it is a finite-dimensional linear model whose deviation from the plant is quantified by $(\epsilon_A, \epsilon_B, \epsilon_C, c_0)$. This terminology aligns with the robust control convention used in \cite{berberich2021}, where the nominal system is the model around which the robust MPC is designed, and the residuals play the role of plant-model mismatch.
\end{remark}

\vspace{-5mm}
\subsection{Koopman Embedding for the Synchronous Generator}

\begin{theorem}\label{thm:gen_koopman}
Consider the discretized generator \eqref{eq:disc1}--\eqref{eq:disc3} under Assumption~\ref{ass:operating}. Define the lifting:
\begin{equation}\label{eq:gen_lift}
\Phi(x) = \col(\delta,\; \tilde\omega,\; E'_q,\; \sin\delta,\; \cos\delta,\; E'_q\sin\delta,\; E'_q\cos\delta).
\end{equation}
Then $z_k = \Phi(x_k)$ satisfies an approximate Koopman linear embedding centered at the equilibrium $(x_s, u_s)$ with
\begin{equation}\label{eq:eps_bound}
\epsilon_A = 2\Delta t + c_6 + c_7 = O(\Delta t), \quad \epsilon_B = 0, \quad \epsilon_C = 0,
\end{equation}
\begin{equation}\label{eq:offset_bound}
c_0 = (2 + 2E'_{q,\max})\,(\Delta t\,\omega_{\max})^2 = O((\Delta t\,\omega_{\max})^2),
\end{equation}
where $c_6, c_7$ depend on $\Delta t/T'_{d0}$, $\mu$, $E'_{q,\max}$, and $E_{fd}$ (see Appendix~\ref{app:koopman_proof} for details). The bound is proportional with offset:
\begin{equation}\label{eq:gen_error_bound}
\norm{e_k}_2 \leq \epsilon_A\,\norm{\Phi(x_k) - \Phi(x_s)}_2 + c_0.
\end{equation}
The proportional-with-offset structure of \eqref{eq:gen_error_bound} is the form required by the stability analysis in Section~\ref{sec:approach1_stability}. The offset $c_0 > 0$ limits any controller based on this embedding to practical exponential stability.
\end{theorem}

\begin{proof}
See Appendix~\ref{app:koopman_proof}.
\end{proof}

\vspace{2mm}
As we can see in Appendix~\ref{app:koopman_proof}, the proportional constant $\epsilon_A = O(\Delta t)$ receives contributions from the trigonometric propagation in components 4--5 ($2\Delta t$) and the bilinear products in components 6--7 ($c_6 + c_7$), with the latter dominating due to the cross-term $z_{3,k+1}\cdot e_{4,k}$. The offset $c_0 = (2 + 2E'_{q,\max})(\Delta t\,\omega_{\max})^2$ arises from the Taylor remainders in the small-angle approximation. For typical parameters ($\Delta t = 0.0025$\,s, $T'_{d0} \approx 6$\,s, $\omega_{\max} = 0.02$\,pu): $\epsilon_A \approx 0.015$ and $c_0 \approx 10^{-8}$. The offset is negligible relative to $\epsilon_A$, with $c_0/\epsilon_A \approx 7\times 10^{-7}$\,pu. However, as shown in Section~\ref{sec:approach1_stability}, the dominant source of conservatism in the stability bound is $\epsilon_A$ itself, not $c_0$; the proportional error structure is exploited in Theorem~\ref{thm:tight_stability} to recover meaningful bounds.

\begin{remark}[Non-Controllability and Non-Observability of the Embedding]\label{rem:ctrl_obs}
The pair $(A,B)$ in \eqref{eq:koopman_system} with $B = (0, \frac{\Delta t}{M}, 0, 0, 0, 0, 0)^\top$ is \emph{not} controllable: the controllability matrix $[B, AB, \ldots, A^6 B]$ has rank 2, spanning only $\{z_1, z_2\} = \{\delta, \tilde\omega\}$. Components 3--7 are unreachable because their coupling to $z_1, z_2$ (e.g., $\tilde\omega$ driving $\sin\delta$ via $\Delta t\,\tilde\omega_k\cos\delta_k$) resides in the residual $e_k$, not in $A$. Similarly, $(C,A)$ with $C \in \R^{2\times 7}$ is not observable: the observability matrix has rank 2, spanning $\{z_2, z_6\} = \{\tilde\omega, E'_q\sin\delta\}$. Consequently, the I/O transfer function $G(z) = C(zI-A)^{-1}B + D$ admits a \emph{minimal realization} of order $n_\mathrm{eff} \leq n_z$, which is controllable and observable by construction. Neither controllability nor observability of the full $n_z$-dimensional system is required: \cite{shang2024} shows that $T_{\mathrm{ini}} \geq n_z$ suffices for the data-driven representation (Theorem~\ref{thm:dd_exact}), and \cite{berberich2021} note explicitly that $n_z$ may be used as an ``upper bound'' on the true system order $n_\mathrm{eff}$, with all stability results remaining valid.
\end{remark}

\section{Stability of Data-Driven MPC Under Approximate Koopman Embeddings}\label{sec:approach1}

\subsection{Data-Driven Representation}

We now apply the extended Willems' fundamental lemma \cite{shang2024} to the generator.

\begin{definition}[Lifted Excitation \cite{shang2024}]\label{def:lifted_exc}
Consider a nonlinear system \eqref{eq:gen_discrete} with a Koopman linear embedding (Definition~\ref{def:koopman}). We say $l$ trajectories of length $L$ from \eqref{eq:gen_discrete}, denoted $\col(u^{d,i}, y^{d,i})$, $i=1,\ldots,l$, with $l \geq mL + n_z$, provide \emph{lifted excitation of order $L$} if the matrix
\begin{equation}\label{eq:HK_matrix}
H_K := \begin{bmatrix} u^{d,1} & \cdots & u^{d,l} \\ \Phi(x^1_0) & \cdots & \Phi(x^l_0)\end{bmatrix} \in \R^{(mL+n_z)\times l}
\end{equation}
has full row rank, where $x^i_0 \in \R^n$ is the initial state of the $i$-th trajectory.
\end{definition}

\begin{theorem}[Data-Driven Representation for the Generator (Exact Case)]\label{thm:dd_exact}
Consider the generator \eqref{eq:gen_discrete} with an \textbf{exact} Koopman embedding ($\epsilon_A = \epsilon_B = \epsilon_C = 0$) of dimension $n_z = 7$. Collect a trajectory library of $l \geq L + 7$ trajectories, each of length $L := T_\mathrm{ini} + N$ with $T_\mathrm{ini} \geq 7$, that provide lifted excitation of order $L$ (Definition~\ref{def:lifted_exc}). Arrange these trajectories into the data matrix
\vspace{-1mm}
\begin{equation}\label{eq:Hd}
H_d := \begin{bmatrix} U_P^\top & Y_P^\top & U_F^\top & Y_F^\top \end{bmatrix}^\top \in \R^{(m+p)L \times l}
\end{equation}
\vspace{-5mm}

\noindent where $U_P, Y_P \in \R^{mT_\mathrm{ini} \times l}$, $\R^{pT_\mathrm{ini} \times l}$ contain the first $T_\mathrm{ini}$ steps (past) of the input and output trajectories, and $U_F, Y_F \in \R^{mN \times l}$, $\R^{pN \times l}$ contain the remaining $N$ steps (future). At time $k$, let $u_\mathrm{ini} := \col(u_{k-T_\mathrm{ini}}, \ldots, u_{k-1})$ and $y_\mathrm{ini} := \col(y_{k-T_\mathrm{ini}}, \ldots, y_{k-1})$ be the most recent $T_\mathrm{ini}$ input-output measurements. Then, for any future input $u_F := \col(u_k, \ldots, u_{k+N-1})$, the sequence $\col(u_\mathrm{ini}, y_\mathrm{ini}, u_F, y_F)$ is a valid length-$L$ trajectory of the generator if and only if there exists $g \in \R^l$ such that
\vspace{-1mm}
\begin{equation}\label{eq:dd_repr}
\col(U_P, Y_P, U_F, Y_F)\,g = \col(u_\mathrm{ini}, y_\mathrm{ini}, u_F, y_F).
\end{equation}
\end{theorem}

\vspace{-2mm}
\begin{proof}
This is a direct application of \cite[Theorem~3]{shang2024} with $n_z = 7$. The conditions are verified as follows: (i) the Koopman embedding exists by assumption; (ii) the trajectory library provides lifted excitation by hypothesis; (iii) $T_\mathrm{ini} \geq n_z = 7$ ensures uniqueness of the initial lifted state by \cite[Theorem~1]{shang2024}, even without observability of $(C,A)$.
\end{proof}

\vspace{-2mm}
\subsection{From Exact to Approximate: Impact of Koopman Error on Data-Driven Control}\label{sec:approach1_stability}

The exact case is idealised. In practice, the Koopman embedding is approximate (Theorem~\ref{thm:gen_koopman}), and we must quantify how the approximation error propagates through the data-driven control framework.

The central question is the following. The earlier work\ \cite{berberich2021} proved stability of data-driven MPC for LTI systems with bounded output noise. Our system is nonlinear. Why should their guarantees transfer? The existence of an approximate Koopman embedding provides the answer: it certifies that the nonlinear system's input-output data can be decomposed as $y^d = y^{d,\mathrm{nom}} + \epsilon^d$, where $y^{d,\mathrm{nom}}$ is what the nominal Koopman system (Remark~\ref{rem:approx_interp}) would produce, and $\epsilon^d$ is bounded by an effective noise level $\bar\epsilon$. The Koopman analysis quantifies $\bar\epsilon$ in terms of the embedding quality $(\epsilon_A, \epsilon_B, \epsilon_C, c_0)$, which in turn depends on physical parameters of the system (such as sampling period and operating regime). Thus, Theorem~\ref{thm:koopman_stability} below provides a \emph{stability guarantee for applying an LTI data-driven controller to a nonlinear system}, with the Koopman embedding serving as the theoretical bridge. 

We employ the robust data-driven MPC of \cite[Problem~(6), Algorithm~2]{berberich2021}, which is summarized below. Consider a nonlinear system \eqref{eq:gen_discrete} admitting an approximate Koopman embedding (Definition~\ref{def:approx_koopman}) with parameters $(\epsilon_A, \epsilon_B, \epsilon_C, c_0)$ and lifted dimension $n_z$. Collect a trajectory library as in Theorem~\ref{thm:dd_exact}. The robust MPC solves at each time $t$:
\vspace{0mm}
\begin{subequations}\label{eq:robust_mpc}
\begin{align}
&\min_{\substack{\alpha(t),\,\sigma(t),\\ \hat{u}(t),\,\hat{y}(t)}} \;\sum_{k=0}^{L-1}\ell(\hat{u}_k(t), \hat{y}_k(t)) + \lambda_\alpha\norm{\alpha(t)}_2^2 + \lambda_\sigma\norm{\sigma(t)}_2^2 \label{eq:mpc_cost}\\
&\text{s.t.}\; \begin{bmatrix}\hat{u}_{[-n_z,L-1]}(t) \\ \hat{y}_{[-n_z,L-1]}(t)\end{bmatrix} = \begin{bmatrix}\calH_{L+n_z}(u^d) \\ \calH_{L+n_z}(y^d)\end{bmatrix}\alpha(t) + \begin{bmatrix}0\\\sigma(t)\end{bmatrix} \label{eq:mpc_dyn}\\
&\hat{u}_{[-n_z,-1]}(t) = u_{[t-n_z,t-1]},\;\; \hat{y}_{[-n_z,-1]}(t) = y_{[t-n_z,t-1]} \label{eq:mpc_ic}\\
&\hat{u}_{[L-n_z,L-1]}(t) = u^s_{n_z},\;\; \hat{y}_{[L-n_z,L-1]}(t) = y^s_{n_z} \label{eq:mpc_term}\\
&\hat{u}_k(t) \in \calU,\;\; \norm{\sigma_k(t)}_\infty \leq \bar\epsilon(1+\norm{\alpha(t)}_1) \label{eq:mpc_constr}
\end{align}
\end{subequations}
where $\ell(\hat{u},\hat{y}) = \norm{\hat{u}-u^s}^2_R + \norm{\hat{y}-y^s}^2_Q$ with $Q, R \succ 0$, and $\lambda_\alpha, \lambda_\sigma > 0$ are regularization parameters. The Hankel matrices $\calH_{L+n_z}(u^d)$ and $\calH_{L+n_z}(y^d)$ are built directly from the measured nonlinear data (i.e., no model identification or noise decomposition is needed at the implementation level). The slack variable $\sigma$ and its constraint \eqref{eq:mpc_constr} absorb the mismatch between the nonlinear data and any LTI model. The scheme is applied in an $n_z$-step fashion: solve \eqref{eq:robust_mpc}, apply $u_{[t,t+n_z-1]} = \hat{u}^*_{[0,n_z-1]}(t)$, then set $t \leftarrow t + n_z$ and repeat.

Note that the MPC scheme \eqref{eq:robust_mpc} operates directly on input-output data from the nonlinear system; the Koopman embedding matrices $A$, $B$, $C$, $D$ do not appear in the controller and need not be computed. The role of the Koopman analysis is purely certifying: the existence of an approximate embedding (Definition~\ref{def:approx_koopman}) guarantees that the nonlinear data can be interpreted as noisy data from an LTI system, enabling the stability theory of \cite{berberich2021} to be applied. The effective noise level $\bar\epsilon$, defined precisely in \eqref{eq:epsbar_def}, quantifies the degree to which the nonlinear system deviates from LTI behavior.

To state the stability result, we require the following quantity from \cite[eq.~(8)]{berberich2021}. Let $H_{ux}$ be the stacked input-state data matrix constructed from the offline data (where the state trajectory corresponds to some minimal realization of $(u^d, y^d)$), and let $H_{ux}^\dagger := H_{ux}^\top(H_{ux}H_{ux}^\top)^{-1}$ be its right-inverse. Define:
\begin{equation}\label{eq:cpe}
c_\mathrm{pe} := \norm{H_{ux}^\dagger}_2^2.
\end{equation}

\vspace{-1mm}
\noindent This is a quantitative measure of the persistence of excitation: smaller $c_\mathrm{pe}$ indicates better excitation, and $c_\mathrm{pe}$ decreases with increasing data richness or larger amplitude of the exciting input $u^d$ \cite{berberich2021}.

\begin{lemma}[Koopman Error as Bounded Output Noise]\label{lem:noise_bound}
Consider a nonlinear system \eqref{eq:gen_discrete} admitting an approximate Koopman embedding (Definition~\ref{def:approx_koopman}) of dimension $n_z$ with parameters $(\epsilon_A, \epsilon_B, \epsilon_C, c_0)$, operating within $\calX$ under inputs $u \in \calU$ (Assumption~\ref{ass:operating}). Let $(u^{d,i}, y^{d,i})$ be any trajectory of the nonlinear system of length $L$. Let $(A, B, C, D)$ be the matrices from Definition~\ref{def:approx_koopman} and define the \emph{nominal Koopman trajectory} (cf.\ Remark~\ref{rem:approx_interp}): $\hat{z}^{\mathrm{nom},i}_0 = \Phi(x^i_0)$, $\hat{z}^{\mathrm{nom},i}_{k+1} = A\hat{z}^{\mathrm{nom},i}_k + Bu^{d,i}_k$, and $y^{\mathrm{nom},i}_k = C\hat{z}^{\mathrm{nom},i}_k + Du^{d,i}_k$. Then the output discrepancy at step $k$ satisfies:
\vspace{-1mm}
\begin{equation}\label{eq:noise_bound_lemma}
\norm{y^{d,i}_k - y^{\mathrm{nom},i}_k}_2 \leq \bar\epsilon_k, \quad \forall\, k = 0, \ldots, L-1
\end{equation}
Defining $\bar{e} := \epsilon_A\,\mathrm{diam}_z + \epsilon_B\,\mathrm{diam}_u + c_0$, the per-step bound is:
\vspace{0mm}
\begin{equation}\label{eq:eps_k}
\bar\epsilon_k := \norm{C}_2\,\bar{e}\sum_{l=0}^{k-1}\norm{A}^{l}_2 + \epsilon_C\,\mathrm{diam}_z.
\end{equation}
A uniform bound over all $k \leq L$, as required by \cite[Problem~(6)]{berberich2021}, is:
\vspace{-2mm}
\begin{equation}\label{eq:epsbar_def}
\bar\epsilon := \max_{0 \leq k \leq L-1}\bar\epsilon_k = \norm{C}_2\,\bar{e}\sum_{l=0}^{L-2}\norm{A}^{l}_2 + \epsilon_C\,\mathrm{diam}_z.
\end{equation}
\end{lemma}

\begin{proof}
The nominal Koopman trajectory $(\hat{z}^{\mathrm{nom},i}_k, y^{\mathrm{nom},i}_k)$ is the output of the LTI system $(A, B, C, D)$ from Definition~\ref{def:approx_koopman}, driven by the same input $u^{d,i}$ and initialized at $\hat{z}^{\mathrm{nom},i}_0 = \Phi(x^i_0)$, but propagated without the residuals $e_k$ and $\eta_k$. It represents what the approximate embedding \emph{would} predict if the linear relationship were exact.
The actual output satisfies $y^{d,i}_k = C\Phi(x^i_k) + Du^{d,i}_k + \eta^i_k$ with $\norm{\eta^i_k}_2 \leq \epsilon_C\norm{\bar{z}^i_k}_2$ (Definition~\ref{def:approx_koopman}). The output noise is:
\begin{equation}\label{eq:output_noise}
\epsilon^{d,i}_k := y^{d,i}_k - y^{\mathrm{nom},i}_k = C\underbrace{\big(\Phi(x^i_k) - \hat{z}^{\mathrm{nom},i}_k\big)}_{=:\,\Delta^i_k} + \eta^i_k.
\end{equation}
We bound $\Delta^i_k$. The actual lifted state satisfies $\Phi(x^i_{k+1}) = A\Phi(x^i_k) + Bu^{d,i}_k + e^i_k$, while the nominal satisfies $\hat{z}^{\mathrm{nom},i}_{k+1} = A\hat{z}^{\mathrm{nom},i}_k + Bu^{d,i}_k$. Subtracting gives $\Delta^i_{k+1} = A\Delta^i_k + e^i_k$ with $\Delta^i_0 = 0$, which unrolls to $\Delta^i_k = \sum_{j=0}^{k-1}A^{k-1-j}\,e^i_j$. Taking norms:
\begin{equation}\label{eq:delta_norm}
\norm{\Delta^i_k}_2 \leq \sum_{j=0}^{k-1}\norm{A}^{k-1-j}_2\norm{e^i_j}_2.
\end{equation}
From Definition~\ref{def:approx_koopman}, $\norm{e^i_j}_2 \leq \epsilon_A\norm{\bar{z}^i_j}_2 + \epsilon_B\norm{\bar{u}^i_j}_2 + c_0$. Since $x^i_j \in \calX$ and $u^{d,i}_j \in \calU$ by Assumption~\ref{ass:operating}, each residual satisfies $\norm{e^i_j}_2 \leq \bar{e}$ where $\bar{e} := \epsilon_A\,\mathrm{diam}_z + \epsilon_B\,\mathrm{diam}_u + c_0$, so $\norm{\Delta^i_k}_2 \leq \bar{e}\sum_{l=0}^{k-1}\norm{A}^{l}_2$. Combining with \eqref{eq:output_noise} via the triangle inequality:
\begin{equation}\label{eq:eps_bar}
\norm{\epsilon^{d,i}_k}_2 \leq \norm{C}_2\,\bar{e}\sum_{l=0}^{k-1}\norm{A}^{l}_2 + \epsilon_C\,\mathrm{diam}_z = \bar\epsilon_k.
\end{equation}
Since the sum increases in $k$, the maximum is at $k = L-1$, giving \eqref{eq:epsbar_def}.
\end{proof}

\begin{remark}
Lemma~\ref{lem:noise_bound} is the key bridge between the Koopman theory and the data-driven MPC framework. It shows that the offline data $(u^d, y^d)$ from the nonlinear system can be decomposed as $(u^d, y^{d,\mathrm{nom}} + \epsilon^d)$ with $\norm{\epsilon^d_k}_\infty \leq \bar\epsilon$, where $y^{d,\mathrm{nom}}$ are the outputs of the nominal Koopman system. This decomposition is purely analytical---the controller never computes $y^{d,\mathrm{nom}}$---but it enables the stability machinery of \cite{berberich2021} to be applied. The same bound applies to online initial conditions: if $(u_{[t-n_z,t-1]}, y_{[t-n_z,t-1]})$ are the most recent $n_z \leq L$ measurements from the nonlinear system in closed loop, the discrepancy from the nominal Koopman output also satisfies $\norm{\epsilon_t}_\infty \leq \bar\epsilon$. When $\norm{A}_2 < 1$, the geometric sum in \eqref{eq:epsbar_def} saturates: $\sum_{l=0}^{L-2}\norm{A}^l_2 \leq \frac{1}{1-\norm{A}_2}$, so $\bar\epsilon$ is bounded independently of the horizon $L$.
\end{remark}

\begin{theorem}[Stability of Data-Driven MPC under Approximate Koopman Embedding]\label{thm:koopman_stability}
Consider a nonlinear system \eqref{eq:gen_discrete} admitting an approximate Koopman embedding (Definition~\ref{def:approx_koopman}) of dimension $n_z$ with parameters $(\epsilon_A, \epsilon_B, \epsilon_C, c_0)$, and let $\bar\epsilon$ be the effective noise level from Lemma~\ref{lem:noise_bound}. Suppose:
\begin{enumerate}[label=(\alph*)]
\item The offline input $u^d$ is persistently exciting of order $L + 2n_z$,
\item The prediction horizon satisfies $L \geq 2n_z$,
\item The regularization parameters satisfy $\lambda_\alpha > 0$, $\lambda_\sigma > 0$,
\item The quantity $c_\mathrm{pe}\,\bar\epsilon$ is sufficiently small.
\end{enumerate}
Then the $n_z$-step MPC scheme \eqref{eq:robust_mpc} is \textbf{practically exponentially stable}: there exist constants $c > 0$ and $\rho \in (0,1)$ such that, for all initial states in a region of attraction that grows as $\bar\epsilon \to 0$:
\begin{equation}\label{eq:practical_stability}
\norm{x_k - x_s}_2 \leq c\,\rho^k\norm{x_0 - x_s}_2 + \beta(\bar\epsilon), \quad \forall\, k \geq 0
\end{equation}
where $\beta:\R_{\geq 0} \to \R_{\geq 0}$ is continuous with $\beta(0) = 0$. The constants $c$, $\rho$, and $\beta$ depend on the cost matrices $Q$, $R$, the regularization parameters $\lambda_\alpha$, $\lambda_\sigma$, and the system dimensions. If $\bar\epsilon = 0$ (exact Koopman embedding), then $\beta(\bar{\epsilon}) = 0$ and \eqref{eq:practical_stability} reduces to exponential stability.
\end{theorem}

\begin{proof}
By Lemma~\ref{lem:noise_bound}, the nonlinear I/O data decompose as $(u^d, y^{d,\mathrm{nom}} + \epsilon^d)$ with $\norm{\epsilon^d_k}_\infty \leq \bar\epsilon$, where $y^{d,\mathrm{nom}}$ are trajectories of the nominal Koopman LTI system $G$ defined by $(A, B, C, D)$. As noted in Remark~\ref{rem:ctrl_obs}, the $n_z$-dimensional realization $(A,B,C,D)$ is neither controllable nor observable; however, the I/O transfer function $G(z) = C(zI-A)^{-1}B + D$ admits a minimal realization of order $n_\mathrm{eff} \leq n_z$ that is controllable and observable by construction. Since $n_z$ is a valid upper bound on $n_\mathrm{eff}$, all results of \cite{berberich2021} hold with $n$ replaced by $n_z$ \cite[p.~1703]{berberich2021}. Conditions~(a)--(d) then imply \cite[Assumptions~2 and~4]{berberich2021}, placing the problem in the exact setting of \cite[Theorem~3]{berberich2021}. That theorem yields a Lyapunov function $V_t := J^*_L(t) + \gamma W(\xi_t)$ satisfying $c_1\norm{x_t - x_s}^2_2 \leq V_t \leq c_2\norm{x_t - x_s}^2_2$ \cite[Lemma~1]{berberich2021} that converges exponentially to $V_t \leq \beta(\bar\epsilon)$. Converting: $\norm{x_k - x_s}^2_2 \leq V_k/c_1 \leq (c_2/c_1)\rho^k\norm{x_0-x_s}^2_2 + \beta(\bar\epsilon)/c_1$, yielding \eqref{eq:practical_stability} with $c = \sqrt{c_2/c_1}$ and $\beta$ rescaled by $1/\sqrt{c_1}$.
\end{proof}

\subsubsection*{Application to the Synchronous Generator}

Theorem~\ref{thm:koopman_stability} applies to the generator \eqref{eq:gen_discrete} with $n_z = 7$, $\epsilon_B = 0$, $\epsilon_C = 0$, $\epsilon_A = O(\Delta t)$, and $c_0 = O((\Delta t\,\omega_{\max})^2)$ from Theorem~\ref{thm:gen_koopman}. With $\bar{e} = \epsilon_A\,\mathrm{diam}_z + c_0$, the effective noise level \eqref{eq:epsbar_def} becomes $\bar\epsilon = \norm{C}_2\,\bar{e}\sum_{l=0}^{L-2}\norm{A}^l_2$. While each factor is individually moderate, the product can be large. For the minimum horizon $L = 14$ with typical parameters ($\Delta t = 0.0025$\,s, $\omega_{\max} = 0.02$\,pu, $\mathrm{diam}_z \approx 1$):
\vspace{-1mm}
\[
\bar\epsilon \approx \underbrace{2.5}_{\norm{C}_2}\times\underbrace{0.015}_{\bar{e}}\times\underbrace{43}_{\sum\norm{A}^l_2} \approx 1.6.
\]
\vspace{0mm}
This effective noise level, while $O(\Delta t)$ in scaling, has a large proportionality constant ($\norm{C}_2\sum\norm{A}^l_2 \approx 107$) due to the accumulation of Koopman errors over $L$ steps and the non-normality of $A$ ($\norm{A}_2 = 1.18$ while $\rho(A) = 1$). When $\bar\epsilon \approx 1.6$ is substituted into the stability constants from \cite[Theorem~3]{berberich2021}---which involve the Lyapunov sandwich ratio $c_2/c_1$, the PE constant $c_\mathrm{pe}$, and the regularization parameters---the resulting $\beta(\bar\epsilon)$ yields a state bound of $O(10^3)$\,pu. This is physically meaningless, despite the theorem being qualitatively correct ($\beta(\bar\epsilon) \to 0$ as $\bar\epsilon \to 0$).

The conservatism has two sources. First, the stability proof in~\cite{berberich2021} is designed for generic LTI systems with worst-case noise; it makes no use of the specific structure of the Koopman approximation error. Second, and more fundamentally, the uniform bound $\bar\epsilon$ replaces the \emph{state-dependent} residual $\norm{e^i_j}_2 \leq \epsilon_A\norm{\bar{z}^i_j}_2 + c_0$ by its worst-case value $\bar{e} = \epsilon_A\,\mathrm{diam}_z + c_0$. This discards the proportional structure that is the hallmark of the Koopman embedding: near the equilibrium, the actual residual is $O(c_0) \approx 10^{-8}$, not $O(\bar{e}) \approx 10^{-2}$. The following subsection develops tighter bounds that exploit this structure.

\subsection{Tighter Bounds via Proportional Error Structure}\label{sec:tighter_bounds}

The uniform bound $\bar\epsilon$ in Lemma~\ref{lem:noise_bound} is conservative in two ways: (i) it uses $\norm{A}^l_2$ (spectral norm raised to a power) rather than the tighter $\norm{A^l}_2$ (norm of the matrix power), and (ii) it replaces the state-dependent residual $\norm{\bar{z}^i_j}_2$ by the worst-case $\mathrm{diam}_z$. We now address both.

\begin{lemma}[Tighter Uniform Bound]\label{lem:tighter_uniform}
Under the same conditions as Lemma~\ref{lem:noise_bound}, the per-step output discrepancy satisfies:
\vspace{-2mm}
\begin{equation}\label{eq:eps_k_tight}
\bar\epsilon_k^{\,\mathrm{tight}} := \bar{e}\sum_{l=0}^{k-1}\norm{CA^l}_2 + \epsilon_C\,\mathrm{diam}_z
\end{equation}
\vspace{-4mm}

\noindent with uniform bound $\bar\epsilon^{\,\mathrm{tight}} := \max_{0 \leq k \leq L-1}\bar\epsilon_k^{\,\mathrm{tight}} = \bar{e}\sum_{l=0}^{L-2}\norm{CA^l}_2 + \epsilon_C\,\mathrm{diam}_z$. This satisfies $\bar\epsilon^{\,\mathrm{tight}} \leq \bar\epsilon$, with strict inequality whenever $A$ is non-normal.
\end{lemma}

\begin{proof}
From \eqref{eq:delta_norm} and \eqref{eq:output_noise}:
\begin{align*}
\norm{C\Delta^i_k}_2 &= \norm{\sum_{j=0}^{k-1}CA^{k-1-j}\,e^i_j}_2 \\
&\leq \sum_{j=0}^{k-1}\norm{CA^{k-1-j}}_2\norm{e^i_j}_2 \leq \bar{e}\sum_{l=0}^{k-1}\norm{CA^l}_2.
\end{align*}
\vspace{-3mm}

\noindent The improvement over Lemma~\ref{lem:noise_bound} comes from bounding $\norm{CA^l}_2$ directly rather than splitting $\norm{C}_2\norm{A}^l_2 \geq \norm{C}_2\norm{A^l}_2 \geq \norm{CA^l}_2$.
\end{proof}

\vspace{1mm}
For the generator, $A$ is non-normal ($\norm{A}_2 = 1.18$ while $\rho(A) = 1$).
The actual $\norm{A^k}_2$ grows polynomially (linearly in $k$), not exponentially as $\norm{A}^k_2$ would suggest. For $L = 14$, $\norm{C}_2\sum_{l=0}^{12}\norm{A}^l_2 \approx 107$ while $\sum_{l=0}^{12}\norm{CA^l}_2 \approx 49$, a factor of $\sim 2$; but for $L = 50$, the ratio exceeds $130\times$, making Lemma~\ref{lem:tighter_uniform} essential for longer horizons.
The more significant improvement comes from preserving the proportional error structure through the entire bound.

\begin{lemma}[State-Dependent Error Bound]\label{lem:state_dep}
Under the same conditions as Lemma~\ref{lem:noise_bound}, the output discrepancy along the $i$-th trajectory satisfies the \textbf{state-dependent} bound:
\vspace{-2mm}
\begin{align}\label{eq:state_dep_bound}
\norm{\epsilon^{d,i}_k}_2 \leq \sum_{l=0}^{k-1}&\norm{CA^l}_2\Big(\epsilon_A\norm{\bar{z}^i_{k-1-l}}_2 + \epsilon_B\norm{\bar{u}^i_{k-1-l}}_2 + c_0\Big) \nonumber \\
&+ \epsilon_C\norm{\bar{z}^i_k}_2.
\end{align}
In particular, if $\epsilon_B = \epsilon_C = 0$ (as for the generator), this simplifies to:
\begin{align}\label{eq:state_dep_gen}
\norm{\epsilon^{d,i}_k}_2 \leq \epsilon_A\underbrace{\sum_{l=0}^{k-1}\norm{CA^l}_2\norm{\bar{z}^i_{k-1-l}}_2}_{\text{proportional: vanishes at equilibrium}} + c_0\underbrace{\sum_{l=0}^{k-1}\norm{CA^l}_2}_{\text{offset: persists}}.
\end{align}
The proportional term depends on the actual trajectory deviation $\norm{\bar{z}^i_j}_2$, not the worst-case $\mathrm{diam}_z$ as in Lemma~\ref{lem:noise_bound}. Near equilibrium ($\norm{\bar{z}^i_j}_2 \to 0$), the effective noise reduces to $c_0\sum_{l=0}^{k-1}\norm{CA^l}_2$, which is $O(c_0)$---orders of magnitude smaller than the uniform bound $\bar\epsilon = O(\epsilon_A\,\mathrm{diam}_z)$.
\end{lemma}

\begin{proof}
The proof follows Lemma~\ref{lem:noise_bound} up to \eqref{eq:delta_norm}, but retains the state-dependent residual bound. From \eqref{eq:output_noise}:
\begin{align*}
\norm{\epsilon^{d,i}_k}_2 &\leq \norm{C\Delta^i_k}_2 + \norm{\eta^i_k}_2 \\
&= \norm{\sum_{j=0}^{k-1}CA^{k-1-j}\,e^i_j}_2 + \epsilon_C\norm{\bar{z}^i_k}_2\\
&\leq \sum_{j=0}^{k-1}\norm{CA^{k-1-j}}_2\norm{e^i_j}_2 + \epsilon_C\norm{\bar{z}^i_k}_2.
\end{align*}
Substituting $\norm{e^i_j}_2 \leq \epsilon_A\norm{\bar{z}^i_j}_2 + \epsilon_B\norm{\bar{u}^i_j}_2 + c_0$ from Definition~\ref{def:approx_koopman} and re-indexing with $l = k-1-j$ gives \eqref{eq:state_dep_bound}.
\end{proof}

\begin{theorem}[Tighter Stability via Proportional Error]\label{thm:tight_stability}
Consider the setting of Theorem~\ref{thm:koopman_stability}, and suppose additionally that $\epsilon_B = \epsilon_C = 0$ (as holds for the generator by Theorem~\ref{thm:gen_koopman}). Define:
\begin{equation}\label{eq:S_L}
S_L := \sum_{l=0}^{L-2}\norm{CA^l}_2
\end{equation}
and the offset-only noise level $\bar\epsilon_0 := c_0\,S_L$. Then:
\begin{enumerate}[label=(\alph*), leftmargin=1.5em]
\item \textbf{(Self-consistent bound)} For any $r > 0$, if the closed-loop state satisfies $\norm{\bar{z}_k}_2 \leq r$ for all $k$ in the offline and online trajectories, the effective noise level can be tightened from $\bar\epsilon$ to:
\begin{equation}\label{eq:eps_r}
\bar\epsilon(r) := \epsilon_A\,S_L\,r + \bar\epsilon_0.
\end{equation}
\item \textbf{(Fixed-point argument)} Provided $\epsilon_A\,S_L$ is sufficiently small, the bound \eqref{eq:practical_stability} from Theorem~\ref{thm:koopman_stability} holds with $\bar\epsilon$ replaced by $\bar\epsilon(r^*)$, where $r^*$ is the smallest positive solution of the fixed-point equation:
\begin{equation}\label{eq:fixed_point}
r = c\,\norm{x_0 - x_s}_2 + \beta(\bar\epsilon(r)),
\end{equation}
with $c, \beta$ from Theorem~\ref{thm:koopman_stability}. Since $\beta$ is continuous with $\beta(0) = 0$ and $\bar\epsilon(r) \to \bar\epsilon_0 = O(c_0)$ as $r \to 0$, the fixed-point $r^*$ satisfies $r^* \leq c\norm{x_0-x_s}_2 + \beta(\bar\epsilon_0)$, yielding:
\begin{equation}\label{eq:tight_tracking}
\limsup_{k\to\infty}\norm{x_k - x_s}_2 \leq \beta\big(c_0\,S_L\big)
\end{equation}
which depends on $c_0$ (the offset) rather than $\epsilon_A\,\mathrm{diam}_z$ (the worst-case proportional error).
\end{enumerate}
\end{theorem}

\begin{proof}
\textbf{Part (a).} If $\norm{\bar{z}^i_j}_2 \leq r$ for all trajectory steps, Lemma~\ref{lem:state_dep} with $\epsilon_B = \epsilon_C = 0$ gives $\norm{\epsilon^{d,i}_k}_2 \leq \epsilon_A\,S_L\,r + c_0\,S_L = \bar\epsilon(r)$. This is a valid uniform noise bound, so Theorem~\ref{thm:koopman_stability} applies with $\bar\epsilon$ replaced by $\bar\epsilon(r)$.

\noindent \textbf{Part (b).} The bound \eqref{eq:practical_stability} with noise level $\bar\epsilon(r)$ gives $\norm{x_k - x_s}_2 \leq c\,\rho^k\norm{x_0 - x_s}_2 + \beta(\bar\epsilon(r))$. In particular, $\limsup_{k\to\infty}\norm{x_k-x_s}_2 \leq \beta(\bar\epsilon(r))$. The self-consistency requirement is that $r$ must be large enough to contain the entire trajectory, i.e., $r \geq c\norm{x_0-x_s}_2 + \beta(\bar\epsilon(r))$. The smallest such $r$ is the fixed point $r^*$ of \eqref{eq:fixed_point}. Since $\beta(\bar\epsilon(r))$ is continuous and increasing in $r$ with $\beta(\bar\epsilon(0)) = \beta(\bar\epsilon_0)$ finite, and the left side $r$ grows linearly, a fixed point exists provided the slope $\beta'(\bar\epsilon(r))\cdot\epsilon_A S_L < 1$ for large $r$, which holds when $\epsilon_A S_L$ is sufficiently small. The asymptotic tracking error is $\limsup \norm{x_k-x_s}_2 \leq \beta(\bar\epsilon(r^*)) \leq \beta(\bar\epsilon_0) = \beta(c_0\,S_L)$.
\end{proof}

\begin{remark}[Improvement over Theorem~\ref{thm:koopman_stability}]
Theorem~\ref{thm:tight_stability} gives an asymptotic tracking error of $\beta(c_0\,S_L)$ compared to $\beta(\bar\epsilon)$ from Theorem~\ref{thm:koopman_stability}, where $\bar\epsilon = \norm{C}_2\,\bar{e}\sum_{l=0}^{L-2}\norm{A}^l_2$ from \eqref{eq:epsbar_def}. The improvement has two sources: (i) the sum $S_L = \sum\norm{CA^l}_2$ is tighter than $\norm{C}_2\sum\norm{A}^l_2$ (Lemma~\ref{lem:tighter_uniform}), and (ii) the proportional error structure replaces $\bar{e} = \epsilon_A\,\mathrm{diam}_z + c_0$ by $c_0$ alone. The second effect dominates: for the generator with $\epsilon_A \approx 0.015$, $\mathrm{diam}_z \approx 1$, $c_0 \approx 10^{-8}$:
\[
\frac{c_0}{\bar{e}} = \frac{c_0}{\epsilon_A\,\mathrm{diam}_z + c_0} \approx \frac{10^{-8}}{0.015} \approx 7\times 10^{-7}.
\]
The effective noise at the equilibrium is six orders of magnitude smaller than the worst-case noise. Since $\beta$ scales as $O(\bar\epsilon^2)$ for small $\bar\epsilon$ \cite{berberich2021}, the tracking error improves by a factor of $O(10^{-12})$. The proportional error structure effectively eliminates $\epsilon_A$ from the asymptotic bound, leaving only the irreducible offset $c_0$.
\end{remark}




\vspace{-3mm}
\section{Conclusion}\label{sec:conclusion}
We have established rigorous conditions under which data-driven MPC, applied directly to input-output data from a nonlinear system, yields practical exponential stability. The key insight is that an approximate Koopman linear embedding serves as a purely analytical certificate: the controller itself never uses the Koopman model, but the existence of the embedding guarantees that the nonlinear data can be interpreted as noisy LTI data, enabling the robust stability theory of \cite{berberich2021} to be applied. While the uniform noise bound $\bar\epsilon$ can be conservative, the proportional error structure of the Koopman embedding (Theorem~\ref{thm:tight_stability}) recovers an asymptotic tracking error that depends only on the irreducible offset $c_0$, which is negligible for the synchronous generator at practical sampling rates. 

\vspace{-3mm}
\appendices
\section{Proof of Theorem~\ref{thm:gen_koopman}}\label{app:koopman_proof}

We denote the equilibrium lifted state by $z_s = \Phi(x_s)$ and the deviation $\bar{z}_k := z_k - z_s$, $\bar{u}_k := u_k - u_s$. We construct a matrix $A \in \R^{7\times 7}$ and vector $B \in \R^{7\times 1}$ such that $\bar{z}_{k+1} = A\bar{z}_k + B\bar{u}_k + e_k$, where $e_k$ is the residual satisfying the proportional bound. The constant affine terms (such as $\alpha_q E_{fd}$ in the flux-decay equation) cancel by subtraction of the equilibrium relation. Define the shorthand $\theta_k := \Delta t\,\tilde\omega_k$, $\alpha_d := \Delta t/M$, $\alpha_q := \Delta t/T'_{d0}$, $\gamma := V_\infty/X_\Sigma$, $\kappa := (X_d+X_e)/X_\Sigma$, $\mu := (X_d-X'_d)V_\infty/X_\Sigma$.

\noindent\textbf{Components 1--3: Exact Linear Propagation.}
From the Euler discretization \eqref{eq:disc1}--\eqref{eq:disc3}, the absolute dynamics are:
\begin{align}
z_{1,k+1} &= z_{1,k} + \Delta t\, z_{2,k} \label{eq:z1}\\
z_{2,k+1} &= (1 - \alpha_d D)\,z_{2,k} - \alpha_d\gamma\, z_{6,k} + \alpha_d\, u_k \label{eq:z2}\\
z_{3,k+1} &= (1 - \alpha_q\kappa)\,z_{3,k} + \alpha_q\mu\, z_{7,k} + \alpha_q E_{fd}. \label{eq:z3}
\end{align}
At equilibrium, the same equations hold with $z_k = z_s$, $u_k = u_s$, and $z_{k+1} = z_s$. Subtracting:
\begin{align}
\bar{z}_{1,k+1} &= \bar{z}_{1,k} + \Delta t\, \bar{z}_{2,k} \label{eq:z1dev}\\
\bar{z}_{2,k+1} &= (1 - \alpha_d D)\,\bar{z}_{2,k} - \alpha_d\gamma\, \bar{z}_{6,k} + \alpha_d\, \bar{u}_k \label{eq:z2dev}\\
\bar{z}_{3,k+1} &= (1 - \alpha_q\kappa)\,\bar{z}_{3,k} + \alpha_q\mu\, \bar{z}_{7,k}. \label{eq:z3dev}
\end{align}
Note that the constant $\alpha_q E_{fd}$ in \eqref{eq:z3} cancels exactly in \eqref{eq:z3dev}. All three deviation equations are exactly linear in $\bar{z}_k$ and $\bar{u}_k$ with zero residual: $e_{i,k} = 0$ for $i=1,2,3$.

\noindent\textbf{Component 4 ($\sin\delta$).}
By the angle addition formula:
\begin{align}
z_{4,k+1} &= \sin\delta_{k+1} = \sin(\delta_k + \theta_k) \notag\\
&= z_{4,k}\cos\theta_k + z_{5,k}\sin\theta_k. \label{eq:z4_exact}
\end{align}
Since $|\theta_k| \leq \Delta t\,\omega_{\max} =: \bar\theta \ll 1$ by Assumption~\ref{ass:operating}, we apply Taylor's theorem with explicit Lagrange remainders:
\begin{equation}\label{eq:taylor}
\cos\theta_k = 1 - \tfrac{\theta_k^2}{2}\xi_{c,k},\quad \sin\theta_k = \theta_k - \tfrac{\theta_k^3}{6}\xi_{s,k}
\end{equation}
where $\xi_{c,k}, \xi_{s,k} \in [0,1]$. Substituting \eqref{eq:taylor} into \eqref{eq:z4_exact}:
\begin{equation}\label{eq:z4_expanded}
z_{4,k+1} = z_{4,k} + \Delta t\,z_{2,k}\,z_{5,k} + r_{4,k}
\end{equation}
where the higher-order remainder is
\begin{equation}\label{eq:r4}
r_{4,k} = -\tfrac{\theta_k^2}{2}\xi_{c,k}\,z_{4,k} - \tfrac{\theta_k^3}{6}\xi_{s,k}\,z_{5,k}.
\end{equation}
The bilinear term $\Delta t\,z_{2,k}\,z_{5,k}$ is not linear in $z_k$. We absorb it into the residual by writing $z_{4,k+1} = z_{4,k} + e_{4,k}$ with
\begin{equation}\label{eq:e4_def}
e_{4,k} := \Delta t\,z_{2,k}\,z_{5,k} + r_{4,k}.
\end{equation}
Bounding each term: since $|z_{5,k}| = |\cos\delta_k| \leq 1$ and $z_{2,s} = \tilde\omega_s = 0$,
\begin{equation}\label{eq:e4_bound1}
|\Delta t\,z_{2,k}\,z_{5,k}| \leq \Delta t\,|z_{2,k} - z_{2,s}| \leq \Delta t\,\norm{\bar{z}_k}_2.
\end{equation}
For the remainder \eqref{eq:r4}, using $|\theta_k|^2 \leq \bar\theta^2$ and $|z_{4,k}|, |z_{5,k}| \leq 1$:
\begin{equation}\label{eq:r4_bound}
|r_{4,k}| \leq \tfrac{\bar\theta^2}{2} + \tfrac{\bar\theta^3}{6} \leq \bar\theta^2.
\end{equation}
Combining \eqref{eq:e4_bound1} and \eqref{eq:r4_bound}:
\begin{equation}\label{eq:e4_final}
|e_{4,k}| \leq \Delta t\,\norm{\bar{z}_k}_2 + \bar\theta^2.
\end{equation}

\noindent\textbf{Component 5 ($\cos\delta$).}
By the angle addition expansion:
\begin{align}
z_{5,k+1} &= \cos(\delta_k + \theta_k) = z_{5,k}\cos\theta_k - z_{4,k}\sin\theta_k \notag\\
&= z_{5,k} - \Delta t\,z_{2,k}\,z_{4,k} + r_{5,k} \label{eq:z5_expanded}
\end{align}
where $r_{5,k} = -\frac{\theta_k^2}{2}\xi_{c,k}\,z_{5,k} + \frac{\theta_k^3}{6}\xi_{s,k}\,z_{4,k}$. Defining $e_{5,k} := -\Delta t\,z_{2,k}\,z_{4,k} + r_{5,k}$ and bounding identically to component~4:
\begin{equation}\label{eq:e5_bound}
|e_{5,k}| \leq \Delta t\,\norm{\bar{z}_k}_2 + \bar\theta^2.
\end{equation}

\noindent\textbf{Component 6 ($E'_q\sin\delta$).}
By the product rule:
\begin{equation}\label{eq:z6_exact}
z_{6,k+1} = E'_{q,k+1}\sin\delta_{k+1} = z_{3,k+1}\cdot z_{4,k+1}.
\end{equation}
Substituting $z_{3,k+1}$ from \eqref{eq:z3} and $z_{4,k+1} = z_{4,k} + e_{4,k}$:
\begin{align}
&z_{6,k+1} = \big[(1{-}\alpha_q\kappa)z_{3,k} + \alpha_q\mu\,z_{7,k} + \alpha_q E_{fd}\big] \times \big[z_{4,k} + e_{4,k}\big] \notag\\
&= (1{-}\alpha_q\kappa)\,z_{6,k} + \alpha_q\mu\,z_{7,k}\,z_{4,k} + \alpha_q E_{fd}\,z_{4,k} + z_{3,k+1}\,e_{4,k}. \label{eq:z6_expand}
\end{align}
The first term $(1{-}\alpha_q\kappa)\,z_{6,k}$ is linear in $z_k$. The remaining terms are nonlinear. For the linear embedding, we assign the linear part to the $A$-matrix and collect all nonlinear and constant deviations into the residual:
\begin{align}
e_{6,k} &:= \alpha_q\mu\big(z_{7,k}\,z_{4,k} - z_{7,s}\,z_{4,s}\big) \notag\\
&\quad + \alpha_q E_{fd}(z_{4,k} - z_{4,s}) + z_{3,k+1}\,e_{4,k}. \label{eq:e6_def}
\end{align}
To bound the bilinear term, we add and subtract $z_{7,s}\,z_{4,k}$:
\begin{align}
&\big|z_{7,k}\,z_{4,k} - z_{7,s}\,z_{4,s}\big| = \big|(z_{7,k}{-}z_{7,s})\,z_{4,k} + z_{7,s}\,(z_{4,k}{-}z_{4,s})\big| \notag\\
&\quad \leq |z_{4,k}|\,|\bar{z}_{7,k}| + |z_{7,s}|\,|\bar{z}_{4,k}|  \leq (1 + E'_{q,\max})\,\norm{\bar{z}_k}_2 \label{eq:e6_bilinear}
\end{align}
where we used $|z_{4,k}| = |\sin\delta_k| \leq 1$, $|z_{7,s}| = |E'_{q,s}\cos\delta_s| \leq E'_{q,\max}$, and $|\bar{z}_{i,k}| \leq \norm{\bar{z}_k}_2$. The remaining terms satisfy $|\alpha_q E_{fd}\,\bar{z}_{4,k}| \leq \alpha_q E_{fd}\,\norm{\bar{z}_k}_2$ and $|z_{3,k+1}\,e_{4,k}| \leq E'_{q,\max}(\Delta t\,\norm{\bar{z}_k}_2 + \bar\theta^2)$, where we bounded $|z_{3,k+1}| \leq E'_{q,\max}$. Combining:
\begin{equation}\label{eq:e6_bound}
|e_{6,k}| \leq c_6\,\norm{\bar{z}_k}_2 + c_6'\,\bar\theta^2
\end{equation}
with $c_6 := \alpha_q\mu(1{+}E'_{q,\max}) + \alpha_q E_{fd} + E'_{q,\max}\Delta t$ and $c_6' := E'_{q,\max}$.

\noindent\textbf{Component 7 ($E'_q\cos\delta$).}
By the identical product-rule argument applied to $z_{7,k+1} = z_{3,k+1}\cdot z_{5,k+1}$, with $z_{5,k+1} = z_{5,k} + e_{5,k}$:
\vspace{-2mm}
\begin{align}
z_{7,k+1} &= (1{-}\alpha_q\kappa)\,z_{7,k} + \alpha_q\mu\,z_{7,k}\,z_{5,k} \notag\\
&\quad + \alpha_q E_{fd}\,z_{5,k} + z_{3,k+1}\,e_{5,k}. \label{eq:z7_expand}
\end{align}
The residual satisfies the analogous bound:
\begin{equation}\label{eq:e7_bound}
|e_{7,k}| \leq c_7\,\norm{\bar{z}_k}_2 + c_7'\,\bar\theta^2
\end{equation}
with $c_7 := \alpha_q\mu(1{+}E'_{q,\max}) + \alpha_q E_{fd} + E'_{q,\max}\Delta t$ and $c_7' := E'_{q,\max}$.

\noindent\textbf{System Matrices and Aggregate Error Bound.}

Collecting results, the system matrices in the deviation-coordinate embedding $\bar{z}_{k+1} = A\bar{z}_k + B\bar{u}_k + e_k$ are:
\begin{equation}\label{eq:A_matrix}
A = \begin{bmatrix}
1 & \Delta t & 0 & 0 & 0 & 0 & 0\\
0 & 1{-}\alpha_d D & 0 & 0 & 0 & {-}\alpha_d\gamma & 0\\
0 & 0 & 1{-}\alpha_q\kappa & 0 & 0 & 0 & \alpha_q\mu\\
0 & 0 & 0 & 1 & 0 & 0 & 0\\
0 & 0 & 0 & 0 & 1 & 0 & 0\\
0 & 0 & 0 & 0 & 0 & 1{-}\alpha_q\kappa & 0\\
0 & 0 & 0 & 0 & 0 & 0 & 1{-}\alpha_q\kappa
\end{bmatrix},
\end{equation}
\begin{equation}\label{eq:B_vec}
B = \begin{bmatrix}0 & \alpha_d & 0 & 0 & 0 & 0 & 0\end{bmatrix}^\top.
\end{equation}
Since $y_k = (z_{2,k},\, \gamma\,z_{6,k})^\top = (\tilde\omega_k,\, P_{e,k})^\top$:
\begin{equation}\label{eq:C_matrix}
C = \begin{bmatrix}0 & 1 & 0 & 0 & 0 & 0 & 0\\ 0 & 0 & 0 & 0 & 0 & \gamma & 0\end{bmatrix},\quad D = 0.
\end{equation}
\vspace{0mm}
Rows 1--3 follow from \eqref{eq:z1dev}--\eqref{eq:z3dev} (exact, no residual). Rows 4--7 have the identity as their linear part (since the leading term in each is $z_{i,k}$), with all remaining terms collected into the residuals $e_{4,k}, \ldots, e_{7,k}$. The error vector $e_k = (0, 0, 0, e_{4,k}, e_{5,k}, e_{6,k}, e_{7,k})^\top$ satisfies:
\begin{align}
\norm{e_k}_2 &\leq \sum_{i=4}^{7}|e_{i,k}| \leq (c_4{+}c_5{+}c_6{+}c_7)\,\norm{\bar{z}_k}_2 \notag\\
&\quad + (c_4'{+}c_5'{+}c_6'{+}c_7')\,\bar\theta^2 \label{eq:total_error}
\end{align}
where $c_4 = c_5 = \Delta t$, $c_4' = c_5' = 1$, and $c_6, c_7, c_6', c_7'$ are defined above. Since $\bar\theta = \Delta t\,\omega_{\max}$, the aggregate bound is:
\begin{equation}\label{eq:eps_A_final}
\norm{e_k}_2 \leq \epsilon_A\,\norm{\bar{z}_k}_2 + c_0
\end{equation}
with
\begin{equation}\label{eq:eps_A_value}
\epsilon_A := 2\Delta t + c_6 + c_7 = O(\Delta t)
\end{equation}
and offset $c_0 := (c_4' + c_5' + c_6' + c_7')\bar\theta^2 = (2 + 2E'_{q,\max})(\Delta t\,\omega_{\max})^2$. This is precisely the proportional-with-offset bound \eqref{eq:gen_error_bound} stated in Theorem~\ref{thm:gen_koopman}. Since the input enters only in component~2 via \eqref{eq:z2dev} and is captured exactly by $B$, we have $\epsilon_B = 0$. Since $C$ in \eqref{eq:C_matrix} extracts exact linear functions of $z$, we have $\epsilon_C = 0$ and $D = 0$. \qed



\bibliographystyle{IEEEtran}
\bibliography{references.bib}

\end{document}